\documentclass[11pt]{article}

\usepackage[margin=1in]{geometry}
\usepackage{amsmath,amsfonts,amssymb,amsthm}
\usepackage{filecontents}
\usepackage{graphicx}
\usepackage{enumerate}
\usepackage{bbm}
\usepackage{verbatim}
\usepackage{hyperref,color}
\usepackage[capitalize,nameinlink]{cleveref}
\usepackage[dvipsnames]{xcolor}
\hypersetup{
	colorlinks=true,
	pdfpagemode=UseNone,
    citecolor=blue,
    linkcolor=blue,
    urlcolor=blue,
	pdfstartview=FitW
}
\usepackage{appendix}
\crefname{appsec}{Appendix}{Appendices}
\usepackage{tikz}

\theoremstyle{plain}
\newtheorem{theorem}{Theorem}

\newtheorem{lemma}[theorem]{Lemma}
\newtheorem{corollary}[theorem]{Corollary}

\newtheorem*{lemmamainthree}{Lemma~\ref{lem:main3}}
\newtheorem*{lemmacoloring}{Lemma~\ref{lem:coloring}}
\newtheorem*{lemmaind}{Lemma~\ref{lem:ind}}

\theoremstyle{definition}
\newtheorem{definition}[theorem]{Definition}

\newtheorem*{assumption*}{Assumption}

\newtheorem{remark}[theorem]{Remark}
\newtheorem{example}[theorem]{Example}

\crefname{lemma}{Lemma}{Lemmas}
\crefname{theorem}{Theorem}{Theorems}
\crefname{definition}{Definition}{Definitions}
\crefname{fact}{Fact}{Facts}
\crefname{claim}{Claim}{Claims}
\crefname{proposition}{Proposition}{Propositions}

\newcommand{\norm}[1]{\left\lVert #1 \right\rVert}

\newcommand{\dist}{\mathrm{dist}}

\newcommand{\fpras}{\mathsf{FPRAS}}

\newcommand{\fptas}{\mathsf{FPTAS}}
\newcommand{\eee}{\mathrm{e}}

\newcommand{\eps}{\varepsilon}
\newcommand{\Dt}{\Delta}
\newcommand{\gm}{\gamma}
\newcommand{\Gm}{\Gamma}

\newcommand{\BB}{\mathbf{B}}
\newcommand{\lambdab}{\boldsymbol{\lambda}}
\newcommand{\Lambdab}{\boldsymbol{\Lambda}}
\newcommand{\alphab}{\boldsymbol{\alpha}}
\newcommand{\betab}{\boldsymbol{\beta}}

\newcommand{\gb}{\mathbf{g}}
\newcommand{\fb}{\mathbf{f}}
\newcommand{\Lb}{\mathbf{L}}
\newcommand{\onesb}{\mathbf{1}}
\newcommand{\Fc}{\mathcal{F}}
\newcommand{\Bc}{\mathcal{B}}
\newcommand{\Pc}{\mathcal{P}}

\newcommand{\polymer}{\mathrm{pmer}}

\def\rb{\ensuremath{\mathbf{r}}}
\def\cb{\ensuremath{\mathbf{c}}}
\def\ones{\ensuremath{\mathbf{1}}}
\newcommand{\emm}{\mathrm{e}}
\newcommand{\TT}{\intercal}

\title{Sampling Colorings and Independent Sets of Random Regular Bipartite Graphs in the Non-Uniqueness Region}
\author{Zongchen Chen\thanks{School of Computer Science, Georgia Institute of Technology, Atlanta, GA 30332, USA.
		Email: \{chenzongchen, vigoda\}@gatech.edu.
		Research supported in part by NSF grant CCF-2007022.}
	\and
	Andreas Galanis\thanks{Department of Computer Science, University of Oxford, Wolfson Building, Parks Road, Oxford, OX1 3QD, UK.
	Email: andreas.galanis@cs.ox.ac.uk.}
	\and
	Daniel \v{S}tefankovi\v{c}\thanks{Department of Computer Science, University of Rochester, Rochester, NY 14627, USA.
		Email: stefanko@cs.rochester.edu.
		Research supported in part by NSF grant CCF-2007287.}
	\and
	Eric Vigoda$^\star$
}
\date{\today}

\begin{document}

\maketitle

\begin{abstract}
For spin systems, such as the $q$-colorings and independent-set models, approximating the partition function in the so-called non-uniqueness region, where the model exhibits long-range correlations, is typically computationally hard for bounded-degree graphs.  We present new algorithmic results for approximating the partition function and sampling from the Gibbs distribution for spin systems in the non-uniqueness region on random regular bipartite graphs.   We give an $\fpras$ for counting $q$-colorings for even $q=O\big(\tfrac{\Delta}{\log{\Delta}}\big)$ on almost every $\Delta$-regular bipartite graph.  This is within a factor $O(\log{\Delta})$ of the sampling algorithm for general graphs in the uniqueness region and 
improves significantly upon the previous best bound of $q=O\big(\tfrac{\sqrt{\Delta}}{(\log\Delta)^2}\big)$ by Jenssen, Keevash, and Perkins (SODA'19).
Analogously, for the hard-core model on independent sets weighted by $\lambda>0$, we present an $\fpras$ for 
estimating the partition function when $\lambda=\Omega\big(\tfrac{\log{\Delta}}{\Delta}\big)$, 
which  improves upon previous results by an $\Omega(\log \Delta)$ factor. Our results for the colorings and hard-core models follow from a general result that applies to arbitrary spin systems. Our main contribution is to show how to elevate probabilistic/analytic bounds on the marginal probabilities for the typical structure of phases on random bipartite regular graphs into efficient algorithms, using the polymer method. We further show evidence that our result for colorings is within a constant factor of best possible using current polymer-method approaches.
\end{abstract}

\thispagestyle{empty}

\newpage

\setcounter{page}{1}

\section{Introduction}

Polymer models have recently been used to obtain algorithms for spin systems in regimes where standard algorithmic tools (such as correlation-decay algorithms or Gibbs sampling/Glauber dynamics) are inefficient. The prototypical class of graphs where polymer models have been applied to are classes of expander and random regular graphs \cite{JKP, Cannon, cluster, BR19,  liao2019counting,  chen2019fast, biclique}, see also \cite{helmuth2019algorithmic,Borgs,flows} for applications on the grid. 

Random bipartite regular graphs are particularly tantalizing \cite{JKP,liao2019counting, biclique}, since on the one hand there is a somewhat standard probabilistic framework to obtain rough analysis estimates for arbitrary spin systems on them (using first/second moment arguments \cite{antiferro}), but on the other hand the corresponding algorithmic framework, and in particular the development of efficient sampling/counting algorithms, is  lacking. 

This paper will focus on finding the algorithmic limits of the polymer method for the two canonical models of interest, the $q$-colorings and the hard-core model (weighted independent sets), though our results apply much more generally as we will detail later. One of the main contributions of this work is to elevate the rough  guarantees obtained by analytic/probabilistic methods into efficient approximate sampling/counting algorithms.

We begin with the colorings problem: given an integer $q\geq 3$ and a graph $G=(V,E)$ of maximum degree $\Delta$,
the goal is to approximate the number of proper $q$-colorings of $G$, and sample a proper $q$-coloring uniformly at random.   For general graphs there is 
an intriguing computational 
phase transition that is conjectured to occur at the statistical physics phase transition for uniqueness/non-uniqueness of the Gibbs measure on
the infinite $\Delta$-regular tree.  When $q\geq\Delta+2$ it is conjectured that the simple single-site update Markov chain known as the Glauber
dynamics is rapidly mixing on any graph of maximum degree $\Delta$ (rapid mixing refers to a convergence rate which is polynomial in $n=|V|$).  
In contrast when $q\leq\Delta$ it is believed that the problem is intractable.

Current bounds are far from resolving this conjecture but have made considerable progress.  On the algorithmic side, recent results establish
$O(n\log{n})$ mixing time of the Glauber dynamics on an $n$-vertex graph of maximum degree $\Delta$ when $q>(11/6-\eps_0)\Delta$ for a
positive constant $\eps_0\approx 10^{-5}$~\cite{BCCPSV,liu,CDMPP19} and on triangle-free graphs when $q>1.764\Delta$~\cite{CLV20,FGYZ21,CGSV21}.
On the negative side, it was shown in~\cite{antiferro} that for {\em even} $q<\Delta$ it is NP-hard to approximate the number of $q$-colorings. The restriction that $q$ is even in this hardness result is rather technical and is a byproduct of a certain maximisation which was carried out in \cite{antiferro} for even $q$.

The above results address the problem on worst-case graphs; in this paper we address the behavior on \emph{typical/random} graphs.  In this vein, random regular {\em bipartite} graphs
are particularly interesting as they manifest the phase transition of regular trees, and consequently they serve as the key gadget in hardness results~\cite{Sly,SlySun,GSV-ising,antiferro,DFJ,CCGL}.
However, standard approximate counting techniques, such as Markov Chain Monte Carlo (MCMC), fail in the non-uniqueness region; e.g., the Glauber dynamics is exponentially slow to converge, with high probability over the choice of the random regular bipartite graph, for even $q<\Delta$~\cite{antiferro}.

Intriguing algorithmic results for the non-uniqueness region of $q\ll\Delta$ on random bipartite graphs were devised using the recently introduced polymer method of~\cite{JKP} and~\cite{helmuth2019algorithmic}. Jenssen, Keevash, and Perkins~\cite{JKP} presented an $\fptas$ for almost every regular bipartite graph
when $q\leq C\tfrac{\sqrt{\Delta}}{(\log \Delta)^2}$ for a constant $C>0$ (see also the independent result of Liao, Lin, Lu, and Mao~\cite{liao2019counting}).  The running
time of these algorithms was improved to $O(n^2(\log n)^3)$ in~\cite{chen2019fast} using a randomized method, see Remark~\ref{rem:faster} below.

For a graph $G=(V,E)$ and an integer $q\geq 3$, the partition function $Z_G$ is the number of $q$-colorings of $G$.  An algorithm $\mathcal{A}$ is an $\fpras$ 
for the partition function on  almost all $\Delta$-regular bipartite graphs if, with probability $1-o(1)$ over a graph $G$ chosen u.a.r. from $n$-vertex $\Delta$-regular bipartite graphs, given $G$, an accuracy $\epsilon>0$, and a tolerance $\delta>0$, the algorithm $\mathcal{A}$ produces in time $poly(n,1/\epsilon,\log(1/\delta))$ an estimate $\hat{Z}$ of the partition function $Z_G$ satisfying  $(1-\epsilon)Z_G\leq \hat{Z}\leq (1+\epsilon)Z_G$ with probability $\geq 1-\delta$.  The algorithm is an $\fptas$ if it achieves $\delta=0$.

Here we present an $\fpras$ for $q$-colorings on almost every regular bipartite graph for even $q=O(\tfrac{\Delta}{\log{\Delta}})$. This improves significantly over the best previous known bound of $q=O\big(\tfrac{\sqrt{\Delta}}{(\log\Delta)^2}\big)$ given in \cite{JKP}, and is within only an $O(\log \Delta)$-factor from the uniqueness/hardness threshold. In fact, we also provide strong evidence that this is the limit of the polymer method up to the implicit constants in the given bounds, see the upcoming Lemma~\ref{lem:fail} for details.
\begin{theorem}\label{thm:main1}
For all even $q\geq 4$ and all $\Delta \ge 100 q\log q$, there is an $\fpras$ for the number of $q$-colorings on almost all $\Delta$-regular bipartite graphs.
\end{theorem}

We provide analogous results for the hard-core model on weighted independent sets.  For a graph $G=(V,E)$, let $\Omega_G$ denote the collection of 
independent sets of $G$.    For a parameter $\lambda>0$, let independent set $\sigma\in\Omega_G$ have weight $w(\sigma)=\lambda^{|\sigma|}$.  The
partition function for the hard-core model on graph $G$ at fugacity $\lambda$ is defined as $Z_G = \sum_{\sigma\in\Omega} w(\sigma)$ and the
Gibbs distribution is $\mu(\sigma) = w(\sigma)/Z_G$. The hard-core model on the infinite $\Delta$-regular tree undergoes a phase transition corresponding to uniqueness vs. non-uniqueness of the infinite-volume Gibbs measure at $\lambda_c(\Delta) = \tfrac{(\Delta-1)^{\Delta-1}}{(\Delta-2)^\Delta}\sim\tfrac{\eee}{\Delta}$.
For any graph $G$ of maximum degree $\Delta$, for all $\lambda<\lambda_c(\Delta)$, the Glauber dynamics mixes in $O(n\log{n})$ time~\cite{CLV20}.
On the other side, when $\lambda>\lambda_c(\Delta)$, the problem of approximating the partition function is NP-hard on $\Delta$-regular graphs~\cite{Sly,SlySun,GSV-ising}. 

For random $\Delta$-regular bipartite graphs, \cite{JKP} presented an $\fptas$ for $\lambda>50\tfrac{(\log\Delta)^2}{\Delta}$ when $\Delta$ is sufficiently large, and \cite{liao2019counting} for $\lambda\geq 1$ and $\Delta\geq 53$,  see also \cite{Cannon, BR19,liulu} for related results on bipartite graphs. We get an improved range of $\lambda=\Omega(\tfrac{\log{\Delta}}{\Delta})$, which is again within an $O(\log \Delta)$-factor from the uniqueness/hardness threshold.
\begin{theorem}\label{thm:main2}
For all $\Delta\geq 53$ and all $\lambda > 100\frac{\log\Delta}{\Delta}$,
there is an $\fpras$ for the partition function of the
hard-core model with parameter $\lambda$ on almost all $\Delta$-regular bipartite graphs.
\end{theorem}
\begin{remark}\label{rem:faster}
In Theorems~\ref{thm:main1} and~\ref{thm:main2}, we can also obtain deterministic approximation schemes ($\fptas$) by applying the interpolation method, analogously to \cite{JKP}. Here, we follow the Markov-chain framework of \cite{chen2019fast}, which provides substantially stronger running time guarantees than those we state for convenience here. In particular, the FPRASes  in Theorems~\ref{thm:main1} and~\ref{thm:main2} run in time $O\big((\tfrac{n}{\epsilon})^2\log^3(\tfrac{n}{\epsilon})\big)$ when the desired accuracy error is not exponentially small (i.e., $\epsilon\geq \emm^{-\Omega(n)}$). Moreover, in the same range of the parameters, we obtain in addition approximate samplers from the Gibbs distribution with analogous running-time guarantees.
\end{remark}

We remark that the condition that $q$ is even in Theorem~\ref{thm:main1} is for the same technical reasons that the hardness results of \cite{antiferro} were obtained for $q$ even which we stated earlier; we conjecture that the result can be extended to odd $q$ and our proof approach extends verbatim (once one has the analogue of the upcoming Lemma~\ref{lem:coloring}). 

In fact, Theorems~\ref{thm:main1} and~\ref{thm:main2} will be proved as special cases of a general algorithmic result that applies to arbitrary spin systems on random bipartite regular graphs. We first introduce general spin systems following the framework of \cite{biclique}. Note that the techniques in \cite{biclique} were targeted to obtain bounds for general spin system and do not yield tight results, e.g., for  colorings the bound obtained therein is roughly $q=O(\Delta^{1/4})$, cf. with the bound on $q$ in Theorem~\ref{thm:main1}. Also, to obtain the result of the hard-core model in Theorem~\ref{thm:main2} we will also need to explicitly account for the presence of external fields, as detailed in the next section.

\section{Proof Outline}
\subsection{Preliminaries: general spin systems and bicliques}

Let $q\geq 2$ be an integer. A general $q$-spin system $(\BB,\lambdab)$ consists of a symmetric interaction matrix $\BB=\{B_{ij}\}_{i,j\in [q]}$, whose entries are between 0 and 1, and an activity vector $\lambdab=\{\lambda_i\}_{i\in [q]}$ with strictly positive entries which are $\leq 1$. Note, that up to normalising, we may  assume that $\BB$ and $\lambdab$ have each at least one entry equal to 1.

For a graph $G=(V,E)$, an assignment $\sigma:V\rightarrow[q]$ has weight $w_G(\sigma)=\prod_{(u,v)\in E}B_{\sigma(u),\sigma(v)}$. The Gibbs distribution is given by $\mu_G(\sigma)=w_G(\sigma)/Z_G$, where $Z_G=\sum_{\sigma:V\rightarrow[q]}w_G(\sigma)$ is the partition function. We let $\Sigma_G$ be the set of all spin assignments. 
For a spin system on a bipartite graph $G$, the following notion of \emph{bicliques} is relevant \cite{biclique,JKP, galanis2016approximately, sampling}. 
\begin{definition}[Biclique]
For a $q$-spin system with interaction matrix $\BB$, we say that a pair $(S,T)$ where $S,T \subseteq [q]$ is a \emph{biclique} if $B_{ij} = 1$ for all $i \in S, j\in T$. A biclique $(S,T)$ is \emph{maximal} if there is \emph{no} other biclique $(S',T')\neq (S,T)$ satisfying $S \subseteq S' \subseteq [q]$ and $T \subseteq T' \subseteq [q]$. 
\end{definition}
\noindent Note that bicliques are defined using only the interaction matrix $\BB$ and do not depend on $\lambdab$.

\begin{example}\label{example}
For the $q$-colorings model, we have that $\BB$ is the $q\times q$ matrix with all ones except on the diagonal where the entries are zero (and $\lambdab$ is the all-ones vector). The bicliques $(S,T)$ are given by pairs of disjoint sets $S,T\subseteq [q]$, whereas maximal bicliques by pairs of $S,T\subseteq[1]$ that form a partition of $[q]$.   For the hard-core model, we have $\BB=\big[\begin{smallmatrix} 1&1\\1&0\end{smallmatrix}\big]$ and $\lambdab=\big[\begin{smallmatrix} 1\\ \lambda\end{smallmatrix}\big]$. Indexing the rows/columns of $\BB$ with $\{0,1\}$ (instead of $\{1,2\}$), the bicliques are $\{(0, 0), (0, 1), (1, 0), (0, 01), (01, 0)\}$ and the maximal bicliques are $\{(0,01), (01,0)\}$. 
\end{example}

\subsection{Our approach: phase vectors and phase maximality}
\label{subsec:our_approach}
Let $(\BB,\lambdab)$ be an arbitrary spin system and $G$  be a $\Delta$-regular bipartite graph, whose vertex set $V$ is partitioned as $(L,R)$ with $|L|=|R|=n$. 
Our approach to obtain approximation algorithms is to consider the likely frequencies of the spins on each side of the graph in the Gibbs distribution of $G$. Adapting methods from \cite{biclique, JKP}, we show that we can obtain efficient approximation schemes for those spin systems where the ``likely'' frequency vectors are captured by maximal bicliques $(S,T)$, see the upcoming Definition~\ref{def:maximality}. The main new ingredient in our work is to give a tight method to study when this  condition is satisfied for general spin systems, which ultimately yields Theorems~\ref{thm:main1} and~\ref{thm:main2} as special cases.

To formalise the above,  for $q$-dimensional probability vectors $\alphab=\{\alpha_i\}_{i\in [q]},\betab=\{\beta_i\}_{i\in [q]}$, we let 
\begin{equation}\label{eq:sigmaalphabeta}
\Sigma^{\alphab,\betab}_G=\{\sigma:V\rightarrow[q]\, \big|\, |\sigma^{-1}(i)\cap L|= n\alpha_i, |\sigma^{-1}(i)\cap R|= n\beta_i\}
\end{equation}
be the set of spin assignments where exactly $n\alpha_i,n\beta_i$ vertices are assigned the spin $i\in [q]$ on $L,R$, respectively.  Denote by $Z_G^{\alphab,\betab}$ the contribution to the partition function from configurations in $\Sigma^{\alphab,\betab}_G$, i.e., $Z_G^{\alphab,\betab}=\sum_{\sigma\in \Sigma^{\alphab,\betab}_G} w_G(\sigma)$.  We will be interested in those pairs $(\alphab,\betab)$ that contribute significantly to the partition function, as detailed below.
\begin{definition}[Phase vectors]
Let $\eta>0$. For a $q$-spin system on an $n$-vertex regular bipartite graph $G$, we say that a pair $(\alphab,\betab)$ of $q$-dimensional probability vectors  is an $\eta$-phase vector of $G$ if $Z_G^{\alphab,\betab}/Z_G\geq \emm^{-\eta n}$. 
\end{definition}

Understanding the phase vectors is in general a hard task. For random bipartite regular graphs,  these have been identified to lie among the set of fixpoints $(\rb,\cb)$ of the following so-called tree recursions on the $\Delta$-regular tree~\cite{antiferro}:
\begin{equation}\label{eq:tr}
r_i \propto \lambda_i \left(\mbox{$\sum_{j\in [q]}$} B_{ij} c_j\right)^{\Delta-1} \mbox{ for $i \in [q]$};
\qquad
c_j \propto \lambda_j \left(\mbox{$\sum_{i\in [q]}$} B_{ij} r_i\right)^{\Delta-1} \mbox{ for $j \in [q]$}.
\end{equation}
The underpinning principle here leading to this correspondence is that the neighbourhood structure of a random $\Delta$-regular bipartite graph is similar to a $\Delta$-regular tree. Nevertheless, identifying the actual phase vectors, even among the finite-set of fixpoints in \eqref{eq:tr}, has turned out to be rather challenging. Even in the canonical case of $q$-colorings, the current best known analysis works for even $q$ and is a result of technically intense arguments.

Before looking into this in more detail, we first explain how to convert the information about phase vectors into algorithms.  Adapting methods from \cite{biclique, JKP}, we show that this is feasible when the phase vectors correspond to maximal bicliques. More precisely, for a non-empty set $S \subseteq [q]$, define the $q$-dimensional probability vector $\gb_S$  whose $i$-th entry is given by $\frac{\lambda_i}{\sum_{j \in S} \lambda_j}$ for $i \in S$, and 0 otherwise. The following notion of ``phase maximality'' will be important in what follows.

\begin{definition}[Phase Maximality]\label{def:maximality}
Let $(\BB,\lambdab)$ be a $q$-spin system and $\Delta\geq 3$. For $\rho>0$ and a set of maximal bicliques $\Bc_\Delta$, we say that the spin system is $\rho$-maximal with respect to $\Bc_\Delta$  if  there is $\eta>0$ such that,  for almost all $\Delta$-regular bipartite graphs, every $\eta$-phase vector $(\alphab,\betab)$ satisfies $\norm{(\alphab,\betab)-(\gb_S,\gb_T)}_\infty \le \rho$ for some maximal biclique $(S,T)\in \Bc_\Delta$. 
\end{definition}



The key new ingredient to prove Theorems~\ref{thm:main1} and \ref{thm:main2} is to establish maximality for the colorings and hard-core models in the corresponding parameter regimes, as detailed in the following theorems.

\newcommand{\statelemmacoloring}{For even $q\geq 4$  and $\Delta \ge 8 q \log \Delta$, the $q$-colorings model is $\tfrac{1}{12\Delta q}$-maximal with respect to the set of bicliques $\Bc_\Delta=\{(S,[q]\backslash S) \, \big| \, |S|=\frac{q}{2}\}$.}
\begin{lemma}\label{lem:coloring}
\statelemmacoloring
\end{lemma}

\newcommand{\statelemmaind}{For $\Delta\geq 50$ and $\lambda \ge \tfrac{50}{\Delta}$, the hard-core model with fugacity $\lambda$ is $\tfrac{1}{24\Delta}$-maximal with respect to the set of bicliques $\Bc_\Delta=\{(0,01),(01,0)\}$.}
\begin{lemma}\label{lem:ind}
\statelemmaind 
\end{lemma}

Previous approaches in \cite{JKP,biclique,liao2019counting}  to establish the analogues of Lemmas~\ref{lem:coloring} and~\ref{lem:ind} used expansion properties of random $\Delta$-regular bipartite graphs which do not however give tight results in terms of the range of the parameters that they apply. Instead, we follow a more direct analytical approach, using the tree-recursions view mentioned in \eqref{eq:tr}, further details are given in Section~\ref{sec:phases} with the final technical bounds obtained in Section~\ref{sec:maximal}. These more precise bounds allow us to push significantly further the applicability of the polymer method, see also the beginning of Section~\ref{sec:alg} for further explanation.

Indeed, we show that  $\rho$-maximality yields approximation algorithms on random $\Delta$-regular bipartite graphs, provided that $\rho$ is sufficiently small and that the weight of configurations corresponding to maximal bicliques is sufficiently big relatively to other type of configurations. To capture the latter condition, recall that the entries of $\BB,\lambdab$ are between 0 and 1, and each of them includes at least one entry equal to 1.  We say that $\BB$ is a $\delta$-matrix for some $\delta\in [0,1)$ if the second largest entry of $\BB$ is $\leq \delta$, and we denote by $\min(\lambdab)$ the minimum entry in $\lambdab$ (note that this is strictly bigger than 0). By applying the polymer method appropriately (inspired by \cite{biclique}), we show the following in Section~\ref{sec:alg}.

\newcommand{\statelemmamainthree}{Let $(\BB,\lambdab)$ be a $q$-spin system, $\Delta \geq 3$ be an integer, and $\rho=\tfrac{1}{12\Delta q}$. Suppose further that $\BB$ is a $\delta$-matrix for some $\delta\in [0,1)$ and that $\Delta(1-\delta)\min(\lambdab)\geq 7q\big(5+\log \tfrac{(q-1)\Delta^3}{\min(\lambdab)}\big)$.

If the spin system is $\rho$-maximal, then there is an $\fpras$ for the partition function for almost all $\Delta$-regular bipartite graphs. In fact, for almost all $\Delta$-regular bipartite graphs, for $\epsilon=\exp(-\Omega(n))$, the algorithm produces an $\epsilon$-estimate for the partition function and an $\epsilon$-sample from the Gibbs distribution in time $O\big((\tfrac{n}{\epsilon})^2(\log\tfrac{n}{\epsilon})^3\big)$.}
\begin{lemma}\label{lem:main3}
\statelemmamainthree
\end{lemma}

Using the above ingredients, we can prove our main Theorems~\ref{thm:main1} and~\ref{thm:main2}.
\begin{proof}[Proof of Theorems~\ref{thm:main1} and~\ref{thm:main2}]
We first prove the result for colorings, Theorem~\ref{thm:main1}. We just need to combine Lemmas~\ref{lem:coloring} and~\ref{lem:main3}. In the setting of Lemma~\ref{lem:main3} and Example~\ref{example}, we have that the interaction matrix for colorings is a $\delta$-matrix for $\delta=0$ and $\min(\lambdab)=1$. Hence, for $\Delta\geq 100q\log q$, we have that  $\Delta(1-\delta)\min(\lambdab)\geq 7q\big(5+\log \tfrac{(q-1)\Delta^3}{\min(\lambdab)}\big)$ as needed. Moreover, Lemma~\ref{lem:coloring} establishes the required $\rho$-maximality that is further needed. Therefore, the conclusion of Lemma~\ref{lem:main3} applies and we obtain the Theorem~\ref{thm:main1}.

The proof for independent sets, Theorem~\ref{thm:main2}, is analogous, by now combining Lemmas~\ref{lem:ind} and~\ref{lem:main3}. We may assume that $\lambda<1$, otherwise the result follows from the FPRAS for $\Delta\geq 53$ in \cite[Theorem 1]{liao2019counting}. In the setting of  Example~\ref{example}, we have that $q=2$, $\delta=0$ and $\min(\lambdab)=\lambda$. Then, for $\lambda>100\tfrac{\log \Delta}{\Delta}$, we have that  $\Delta(1-\delta)\min(\lambdab)\geq 7q\big(5+\log \tfrac{(q-1)\Delta^3}{\min(\lambdab)}\big)$, and the result follows analogously to above. 
\end{proof}

Finally, as mentioned in the introduction, we give evidence that the bounds on $q$ in Theorem~\ref{thm:main1} capture the limit of the polymer method for colorings, by showing that maximality fails when we go beyond the relevant range (note, some form of maximality is either implicitly or explicitly shown in all previous works on the problems).
\begin{lemma}\label{lem:fail}
For all even $q\geq 4$ and $\Delta=O(q \log q)$, for the $q$-colorings model, $O(\tfrac{1}{\Delta q})$-maximality fails with respect to any set of bicliques on almost all $\Delta$-regular bipartite graphs. 
\end{lemma}
We note that Lemma~\ref{lem:fail} does not exclude the possibility of some exotic polymer model that can perhaps break the barrier therein. It does show however that the current approach cannot go substantially beyond the guarantee in Theorem~\ref{thm:main1}, and at the very least some major refinement of the framework will be needed. We conjecture that a similar barrier applies for the result of Theorem~\ref{thm:main2}, though here the bottleneck is in Lemma~\ref{lem:main3}. More precisely, for $\lambda=O(\tfrac{\log \Delta}{\Delta})$ in the non-uniqueness region,  polymers can be of size roughly $n^{\Omega(1)}$, which is in contrast to what happens when the polymer method applies (where the size of polymers turns out to be logarithmic in $n$).

\section{Phase vectors on random bipartite regular graphs}\label{sec:phases}
Let $(\BB,\lambdab)$ be a $q$-spin system. In this section, we use results from \cite{antiferro} to pinpoint the phase vectors on random $\Delta$-regular bipartite graphs, and give a sufficient condition to conclude maximality (Corollary~\ref{lem:phases}). We will invoke this in Section~\ref{sec:maximal} to prove Lemmas~\ref{lem:coloring} and~\ref{lem:ind}.

For $q$-dimensional probability vectors $\rb,\cb$, we will consider the function 
\[\Phi_{\BB,\lambdab,\Delta}(\rb,\cb)=  \frac{\rb^\TT \BB \cb}{\|\Lambdab^{-1}\rb\|_p \|\Lambdab^{-1}\cb\|_p},\]
where $p=\tfrac{\Delta}{\Delta-1}$  and $\Lambdab$ is the $q\times q$ diagonal matrix whose $i$-th diagonal entry is equal to $\lambda_i^{1/\Delta}$. We will be interested in the maximizers $(\rb,\cb)$ of $\Phi$.
\begin{lemma}\label{lem:maxima}
Suppose that the interaction matrix $\BB$ is ergodic, i.e., irreducible and aperiodic. Then, the maximizers of $\Phi_{\BB,\lambdab,\Delta}(\rb,\cb)$ are fixpoints of the tree recursions~\eqref{eq:tr}.
\end{lemma}
\begin{proof}
The proof follows by a relatively standard Lagrange multiplier argument. The assumption that $\BB$ is ergodic is needed to exlcude maximizers at the boundary, i.e., that some entry of $\rb,\cb$ is equal to zero. A closely related argument in the case $\lambdab=\ones$ can be found in \cite[Lemma 4.11]{antiferro}.
\end{proof}

 Let  $\Lb$ denote the matrix $\big\{\tfrac{B_{ij}r_ic_j}{\sqrt{r_i'c_j'}}\big\}_{i,j\in [q]}$, where $r_i':=r_i(\sum_{j\in [q]} B_{ij} c_j)$ for $i\in [q]$ and $c_j':=c_j(\sum_{i\in [q]}B_{ij}r_i)$ for $j\in [q]$. A maximiser $(\rb,\cb)$ of $\Phi$ is called Hessian dominant in \cite{antiferro} if the eigenvalues of the matrix $\Lb$ apart from the largest (which is equal to 1) are less in absolute value than $\tfrac{1}{\Delta-1}$.  Let $\fb: \rb\mapsto \alphab$ be the map given by $\alpha_i=(\lambda^{-1/\Delta}_ir_i/\|\Lambdab^{-1}\rb\|_p)^p$ for $i\in[q]$.

\begin{lemma}[\cite{antiferro}]\label{lem:hessiandominant}
Let $\Delta\geq 3$ be an integer and consider a $q$-spin system $(\BB,\lambdab)$. Suppose that all the maximizers of $\Phi_{\BB,\lambdab,\Delta}$ are Hessian dominant. Then, for every $\kappa>0$, there is $\eta>0$ such that for almost all $\Delta$-regular bipartite graphs, every $\eta$-phase vector $(\alphab,\betab)$ satisfies $\|(\alphab,\betab)-(\alphab^*,\betab^*)\|_\infty\leq \kappa$, where $(\alphab^*,\betab^*)=(\fb(\rb^*),\fb(\cb^*))$ and $(\rb^*,\cb^*)$ is a maximizer of $\Phi_{\BB,\lambdab,\Delta}$.
\end{lemma}
\begin{proof}
The lemma is proved in \cite[Section 6.4.1]{antiferro} for $\lambdab=\onesb$. To extend to general $\lambdab$, consider the spin system $(\widehat{\BB},\widehat{\lambdab})$ where $\widehat{\BB}=\Lambdab\BB\Lambdab$ and $\widehat{\lambdab}=\ones$.  Note, on $\Delta$-regular bipartite graphs and arbitrary $\eta>0$, an $\eta$-phase vector $(\alphab,\betab)$ of the spin system $(\BB,\lambdab)$ is also an $\eta$-phase vector of the spin system with interaction matrix $\widehat{\BB}$ and activity vector $\mathbf{1}$, and vice versa. Moreover, the maximizers $(\rb^*,\cb^*)$ of $\Phi=\Phi_{\BB,\lambdab,\Delta}(\rb,\cb)$ are in 1-1 correspondence with the maximizers $(\widehat{\rb}^*,\widehat{\cb}^*)$  of $\widehat{\Phi}=\widehat{\Phi}_{\widehat{\BB},\widehat{\lambdab},\Delta}$  via the relation $(\rb^*,\cb^*)=(\Lambdab\widehat{\rb}^*,\Lambdab \widehat{\cb}^*)$. Note also that $(\rb^*,\cb^*)$ is Hessian dominant for $\Phi$ iff $(\widehat{\rb}^*,\widehat{\cb}^*)$ is Hessian dominant for $\widehat{\Phi}$, therefore establishing the result for general $\lambdab$. 
\end{proof}
Using Lemma~\ref{lem:hessiandominant}, we obtain the following using the definition of maximality (cf. Definition~\ref{def:maximality}).
\begin{corollary}\label{lem:phases}
Let $(\BB,\lambdab)$ be a $q$-spin system and $\Delta \geq 3$ be an integer. Suppose that there is a set of maximal bicliques $\Bc_\Delta$ such that all maximizers $(\rb^*,\cb^*)$ of $\Phi_{\BB,\lambdab,\Delta}$ are Hessian dominant and satisfy $\|(\alphab^*,\betab^*)-(\gb_S,\gb_T)\|_\infty\leq \tfrac{1}{15\Delta q}$ for some maximal biclique $(S,T)\in \Bc_\Delta$, where $(\alphab^*,\betab^*)=(\fb(\rb^*),\fb(\cb^*))$. Then, the spin system is $\tfrac{1}{12\Delta q}$-maximal with respect to $\Bc_\Delta$.
\end{corollary}

\section{Algorithms from maximality: Proof of Lemma~\ref{lem:main3}}\label{sec:alg}
Let $\Delta\geq 3$ be an integer, and $(\BB,\lambdab)$ be a $q$-spin system, which is $\rho$-maximal for $\rho=\tfrac{1}{12\Delta q}$.   Consider also a bipartite graph $G=(V,E)$ with vertex bipartition $(L, R)$ and $|L|=|R|=n$. The following expansion property of sets $U\subseteq V$ in random regular bipartite graphs relaxes the previous expansion properties that were used in \cite{JKP,liao2019counting} which needed to consider bigger sets $U$; instead, whenever the spin system is $\tfrac{1}{12\Delta q}$-maximal, we only need to consider sets $U$ with size roughly $\tfrac{1}{\Delta}|V|$, whose expansion is $\Omega(\Delta)$.  For a set $U\subseteq V$, we use $\partial U$ to denote the vertices in $G$ which have a neighbor in $U$ but do not belong to $U$, and by $U^+$ the set $U \cup \partial U$.

\begin{lemma}\label{lem:expansion}
Let $\Delta\geq 3$ be an integer. For almost all $\Delta$-regular bipartite graphs $G=(V,E)$ with bipartition $(L,R)$, the following expansion properties hold: 
\begin{enumerate}
\item every set $U\subseteq V$ with $|U\cap L|\leq \tfrac{1}{3\Delta}|L|$ and $|U\cap R|\leq \tfrac{1}{3\Delta}|R|$  satisfies $|U^+|\geq \tfrac{\Delta-1}{2}|U|$.
\item every set $U\subseteq V$ with $|U\cap L|\leq \tfrac{1}{6\Delta}|L|$ and $|U\cap R|\leq \tfrac{1}{6\Delta}|R|$  satisfies $|\partial U|\geq \tfrac{\Delta}{7}|U|$.
\end{enumerate}
\end{lemma}
\begin{proof}
For the first item, consider a subset $U\subseteq V$ with $|U\cap L|\leq \tfrac{1}{3\Delta}|L|$ and $|U\cap R|\leq \tfrac{1}{3\Delta}|R|$. We will show that 
\begin{equation}\label{eq:expansion1}
|\partial (U \cap L)|\geq \tfrac{\Delta-1}{2}|U\cap L| \mbox{ and } |\partial (U \cap R)|\geq \tfrac{\Delta-1}{2}|U\cap R|.
\end{equation}
From this, we obtain that $|U^+|=|U\cup \partial U|\geq |\partial (U \cap L)|+|\partial (U \cap R)|\geq \tfrac{\Delta-1}{2}|S|$. To verify \eqref{eq:expansion1}, we use a sufficient condition due to Bassalygo \cite{bassalygo1981}, see also \cite[Theorem 22]{JKP}. Namely, for $a=\tfrac{1}{3\Delta}$, $b= \tfrac{\Delta-1}{2}$ and $H(x)=-x \log_2(x)-(1-x)\log_2(1-x)$, we check that $\Delta>\frac{H(a)+H(a b)}{H(a)- a b H(1/b)}$, which indeed holds for all $\Delta\geq 3$.

The proof of the second item is analogous. Consider a subset $U\subseteq V$ with  $|U\cap L|\leq \tfrac{1}{6\Delta}|L|$ and $|U\cap R|\leq \tfrac{1}{6\Delta}|R|$. We will show that 
\begin{equation}\label{eq:expansion2}
|\partial (U \cap L)|\geq \big(\tfrac{\Delta}{7}+1\big)|U\cap L| \mbox{ and } |\partial (U \cap R)|\geq \big(\tfrac{\Delta}{7}+1\big)|U\cap R|.
\end{equation}
From this, we obtain that $|\partial U|\geq |\partial (U \cap L)|+|\partial (U \cap R)|-|U|\geq \tfrac{\Delta}{7}|U|$. The proof of \eqref{eq:expansion2} is by verifying again the same condition as above, now for the values $a=\tfrac{1}{6\Delta}$ and $b= \tfrac{\Delta}{7}+1$.
\end{proof}

Following \cite{biclique}, we will define a polymer model corresponding to a biclique $(S,T)$ of the spin system.  Let $G^3$ be the graph on vertex set $V$ where two vertices $u,v$ are adjacent iff $\dist(u,v)\le 2$. A subset $U\subseteq V$ of vertices is said to be $G^3$-connected if the induced subgraph $G^3[U]$ is connected. A polymer $\gamma=(V_\gm,\sigma_\gm)$ consists of a subset of vertices of $G$, $V_\gamma$, which is $G^3$ connected, and a spin assignment on $V_\gamma$,  $\sigma_\gamma:V_\gamma\rightarrow [q]$, such that every vertex in $V_\gamma \cap L$ gets a spin in $[q]\backslash S$ and every vertex in $V_\gm \cap R$ gets a spin in $[q]\backslash T$. Two polymers $\gm_1,\gm_2$ are compatible (written as $\gm_1 \sim \gamma_2$) if and only if $\dist(\gm_1,\gm_2) > 3$, i.e., $\gm_1 \cup \gm_2$ is not $G^3$-connected. 

The size of a polymer $\gm$, denoted by $|\gm|$, is the number of vertices it contains. We use $E_\gamma$ to denote the edges of $G$ whose both endpoints lie in $\gamma$, $\partial V_\gamma$ to denote the vertices in $G$ which have a neighbor in $V_\gamma$ but do not belong to $V_\gm$, and by $V_\gm^+$ the set $V_\gm \cup \partial V_\gm$.  For a polymer $\gm$, the weight $w^{S,T}_G(\gm)$ of the polymer is given by 
\begin{equation}\label{eq:wSTg}
w^{S,T}_G(\gm)=\frac{\prod_{u\in V_\gamma}\lambda_{\sigma_\gm(u)}\prod_{(u,v)\in E_\gm}B_{\sigma_\gm(u),\sigma_\gm(v)}\prod_{u\in \partial V_\gamma}F_u}{\big(\sum_{i\in S} \lambda_i)^{|V_\gm^+\cap L|}\big(\sum_{j\in T} \lambda_j\big)^{|V_\gm^+\cap R|}},
\end{equation}
where
\begin{equation}\label{eq:Fu}
F_u=\sum_{i\in S}\lambda_i\prod_{v\in V_\gm\cap \partial u}B_{i,\sigma_\gm(v)} \mbox{ if } u\in \partial V_\gm\cap L,\quad
F_u=\sum_{j\in T}\lambda_j\prod_{v\in V_\gm\cap \partial u}B_{j,\sigma_\gm(v)} \mbox{ if } u\in \partial V_\gm\cap R.
\end{equation}

Let $\mathcal{P}^{S,T}_G$ be the set of all polymers $\gm=(V_\gm,\sigma_\gm)$ with $|V_\gm|\leq 2q \rho n=\tfrac{n}{6\Delta}$. A configuration $\Gm=(V_\Gm, \sigma_\Gm)$ of polymers is a collection of mutually compatible polymers $\gamma_1,\hdots, \gamma_k\in \mathcal{P}^{S,T}_G$ with $V_\Gamma=\cup_{t\in [k]} V_{\gm_t}$ and $\sigma_\Gm$ the spin assignment on $V_\Gm$ which agrees with $\sigma_{\gm_t}$ on $V_{\gm_t}$ for each $t\in [k]$. Let $\Omega^{S,T}_G$ be the set of all possible configurations $\Gamma$. The size of a configuration is $|\Gm| = \sum_{\gm\in\Gm} |V_\gm|$. 

\begin{lemma}\label{lem:configupper}
Every configuration $\Gamma$ satisfies $|V_\Gamma|\leq 12n/\Delta$.
\end{lemma}
\begin{proof}
Suppose that there exists a configuration $\Gamma$ with $|V_\Gamma|> 12n/\Delta$. Then, we can extract greedily disjoint configurations $\Gm_1,\hdots, \Gm_{36}\subseteq \Gm$ (which are a collection of polymers belonging to $\Gm$) such that $\tfrac{n}{6\Delta}< |\Gm_i|\leq\tfrac{n}{3\Delta}$. By Lemma~\ref{lem:expansion}, we have that $|V_{\Gm_i}^+|\geq \tfrac{\Delta-1}{2}|\Gamma_i|>\frac{n}{6\Delta}\tfrac{\Delta-1}{2}$ and therefore $\sum^{36}_{t=1}|V_{\Gm_i}^+|>  \frac{6n}{\Delta}\tfrac{\Delta-1}{2}\geq 2n$. Therefore, since $G$ has $2n$ vertices, the sets $V_{\Gm_1}^+,\hdots, V_{\Gm_{36}}^+$ cannot be pairwise disjoint, contradicting the fact that the configuration $\Gm$ consists of pairwise compatible polymers.
\end{proof}

The weight $w^{S,T}_G(\Gm)$ of a configuration $\Gm$ is given by the product of the weights of the polymers that $\Gm$ consists of. We define the partition function of the polymer model as
\[Z^{S,T}_G=\mbox{$\sum_{\Gm\in \Omega^{S,T}_G}$}\, w^{S,T}_G(\Gm), \mbox{ and its Gibbs distribution by } \mu^{S,T}_G(\Gamma)=w^{S,T}_G(\Gm)/Z^{S,T}_G \mbox{ for $\Gm\in \Omega^{S,T}_G$}.\]
Finally, we let $Z^{\polymer}_G=\sum_{(S,T)\in \Bc_\Delta}\,\big(\mbox{$\sum_{i\in S}$}\,\lambda_i\big)^n \big(\mbox{$\sum_{j\in T}$}\,\lambda_j\big)^n Z^{S,T}_G$.

\begin{lemma}\label{lem:estimate}
Let $\Delta\geq3$ be an integer, and $(\BB,\lambdab)$ be a $q$-spin system which is $\tfrac{1}{12\Delta q}$-maximal with respect to a set of maximal bicliques $\Bc_\Delta$. Suppose further that $\Delta \min(\lambdab)\geq 15q$. Then, there is $\epsilon=\emm^{-\Omega(n)}$ such that, for almost all $\Delta$-regular bipartite graphs $G$ with $n$ vertices on each part,  it holds that $(1-\epsilon)Z_G\leq Z^{\polymer}_G\leq (1+\epsilon) Z_G$.
\end{lemma}
\begin{proof}
By the $\tfrac{1}{12\Delta q}$-maximality of the spin system with respect to $\Bc_\Delta$ (cf. Definition~\ref{def:maximality}),  there is an $\eta>0$ such that  for almost all $\Delta$-regular graphs $G$,  every $\eta$-phase vector $(\alphab,\betab)$ of $G$ belongs to 
\[\Fc_\Delta:=\Big\{(\alphab,\betab)\,\Big|\, \| (\alphab,\betab)-(\gb_S,\gb_T)\|_\infty \le \tfrac{1}{12\Delta q} \mbox{ for some maximal biclique } (S,T)\in \Bc_\Delta\Big\}.\] 
Let $\Sigma_G^{\max}=\{\sigma \mid \sigma\in \Sigma^{\alphab,\betab}_G \mbox{ for some } (\alphab,\betab)\in \Fc_\Delta\}$ where, recall from \eqref{eq:sigmaalphabeta}, that $\Sigma^{\alphab,\betab}_G$ is the set of spin assignments where exactly $n\alpha_i,n\beta_i$ vertices are assigned the spin $i\in [q]$ on $L,R$, respectively.

We first show the lower bound on $Z^{\polymer}_G$. Consider the polymer model corresponding to a maximal biclique $(S,T)\in \Bc_\Delta$. Every configuration $\Gm\in \Omega^{S,T}_G$ maps to a set of spin assignments 
\[\Sigma^{S,T}_G(\Gm)=\{\sigma:V\rightarrow [q]\mid \sigma(V_\Gm)=\sigma_\Gm,\, \sigma(L\backslash V_\Gm)\subseteq S,\, \sigma(R\backslash V_\Gm)\subseteq T\},\] where recall that $\sigma_\Gm$ is a spin assignment on $V_\Gm$ that satisfies $\sigma_\Gm(V_\Gm\cap L)\subseteq [q]\backslash S$ and $\sigma_\Gm(V_\Gm\cap R)\subseteq [q]\backslash T$.  Therefore, for distinct $\Gm,\Gm'\in  \Omega^{S,T}_G$ we have that the sets $\Sigma^{S,T}_G(\Gm)$ and $\Sigma^{S,T}_G(\Gm')$ are disjoint. Let $\Sigma^{S,T}_G=\bigcup_{\Gm\in \Omega^{S,T}_G} \Sigma^{S,T}_G(\Gm)$. Using that configurations $\Gamma$ consist of disjoint $G^3$-connected sets, we obtain that the aggregate weight $\sum_{\sigma\in \Sigma^{S,T}_G(\Gm)}w_G(\sigma)$ equals $\big(\mbox{$\sum_{i\in S}$}\,\lambda_i\big)^n \big(\mbox{$\sum_{j\in T}$}\,\lambda_j\big)^n w^{S,T}_G(\Gm)$ (see for example \cite[Lemma 17]{biclique}), and therefore 
\[Z^{S,T}_G=\sum_{\sigma\in \Sigma^{S,T}_G}w_G(\sigma).\]
Moreover, note that for $(\alphab,\betab)$ with $\| (\alphab,\betab)-(\gb_S,\gb_T)\|_\infty \le \tfrac{1}{12\Delta q}$, the number of vertices in $L$ that do not get a spin in $S$ is at most $\tfrac{n}{12\Delta}$, and similarly for vertices in $R$ that do not get a spin in $T$, for a total of $\tfrac{n}{6\Delta}$ vertices, giving that $\Sigma^{\alphab,\betab}_G\subseteq  \Sigma^{S,T}_G$. Observe now that every $(\alphab,\betab)\notin\Fc_\Delta$ is not an $\eta$-phase vector and therefore $Z^{\alphab,\betab}_G\leq \emm^{-\eta n} Z_G$. There are at most $n^{2q}$ such pairs with $n\alphab,n\betab\in\mathbb{Z}^q$ and therefore, combining the above, it follows that 
\[Z_G-Z^{\polymer}_G\leq \sum_{(\alphab,\betab)\notin\Fc_\Delta} Z^{\alphab,\betab}_G\leq n^{2q} \emm^{-\eta n} Z_G\leq \emm^{-\Omega(n)}Z_G,\] showing that $Z^{\polymer}_G\geq (1-\emm^{-\Omega(n)})Z_G$.

We next show the upper bound on $Z^{\polymer}_G$. Consider $\Sigma^{\mathrm{overlap}}_G=\bigcup_{(S,T)\neq (S',T')\in \Bc_\Delta} (\Sigma^{S,T}_G\cap \Sigma^{S',T'}_G)$.    We will show shortly that $\Sigma^{\mathrm{overlap}}_G\subseteq \Sigma_G\backslash \Sigma_G^{\max}$. Assuming this for the moment, we conclude the proof by noting first that for $(\alphab,\betab)$ which is not an $\eta$-phase vector it holds that $Z^{\alphab,\betab}_G/Z_G<\emm^{-\eta n}$. Therefore we obtain that the aggregate weight of spin assignments in $\Sigma^{\mathrm{overlap}}_G$ is at most $n^{2q} \emm^{-\eta n}Z_G=\emm^{-\Omega(n)}Z_G$, yielding that $Z_G\geq (1-\emm^{-\Omega(n)})Z^{\polymer}_G$.

It remains to prove that $\Sigma^{\mathrm{overlap}}_G\subseteq \Sigma_G\backslash \Sigma_G^{\max}$. For the sake of contradiction, suppose otherwise. Then there exists a spin assigment $\sigma$, distinct bicliques $(S,T), (S',T')\in \Bc_\Delta$, and  a biclique $(S^*,T^*)\in \Bc_\Delta$ such that $\sigma\in \Sigma^{S,T}_G\cap \Sigma^{S',T'}_G$ and  $\sigma\in  \Sigma^{\alphab,\betab}_G$ for some $\| (\alphab,\betab)-(\gb_{S^*},\gb_{T^*})\|_\infty \le \tfrac{1}{12\Delta q}$. Since $(S,T)$ and $(S',T')$ are distinct and maximal, we may assume w.l.o.g. have that $S\neq S^*$ and $T\neq T^*$.  Since $(S^*,T^*)$ is maximal, it cannot be the case that $S^*\subseteq S$ and $T^*\subseteq T$, so  assume w.l.o.g. that $i\in S^*\backslash S$.  Let $n_i$ be the vertices in $L$ that have the spin $i$ under $\sigma$. Since $\sigma\in \Sigma^{S,T}_G(\Gm)$ for some $\Gm\in \Omega^{S,T}_G$ and $i\notin S$, from Lemma~\ref{lem:configupper} we have that $n_i\leq |V_{\Gm}|\leq 12n/\Delta$. Then, using the assumption $\Delta \min(\lambdab)\geq 15q$ and the fact that the entries of $\lambdab$ are $\leq 1$, we have the crude bound $\frac{\lambda_i}{\sum_{i'\in S^*}\lambda_{i'}}\geq \min (\lambdab)/q\geq 15/\Delta$, and therefore $|\frac{\lambda_i}{\sum_{i'\in S^*}\lambda_{i'}}-\tfrac{n_i}{n}|\geq \tfrac{3}{\Delta }>\tfrac{1}{12\Delta q}$  contradicting the choice of $(S^*,T^*)$. 
\end{proof}

We are now ready to prove Lemma~\ref{lem:main3}, which we restate here for convenience. The proof uses the Markov chain approach for studying polymer models in \cite{chen2019fast}, as employed for general spin systems in \cite{biclique}.
\begin{lemmamainthree}
\statelemmamainthree
\end{lemmamainthree}
\begin{proof}
The main ingredient that we need to check is  the so-called polymer sampling condition \cite[Definition 4]{chen2019fast} for each polymer model defined by bicliques $(S,T)\in \Bc_{\Delta}$; this gives an $\epsilon$-counting algorithm for $Z^{S,T}_G$ and an $\epsilon$-sampling algorithm for $\mu^{S,T}_G$ with the desired guarantees.  The estimates in Lemma~\ref{lem:estimate} then yield that these algorithms can be extended to algorithms for $Z_G$, by the argument  in \cite[Proof of Theorem 3]{biclique}. 

The polymer sampling condition captures that the weight of the polymers as a function of their size decays exponentially relatively to the growth rate of the number of polymers (containing a vertex); in this case, since we are working with $G^3$ whose degree is bounded by $\Delta^3$, the condition we need to check, cf. \cite[Definition 4]{chen2019fast}, is that $w^{S,T}_G\leq \emm^{-\tau |\gamma|}$ for some constant $\tau\geq 5+3\log((q-1)\Delta^3)$.

Let $\gamma=(V_\gm,\sigma_\gm)\in \mathcal{P}^{S,T}_G$. Since the entries of $\BB,\lambdab$ are $\leq 1$, we have from \eqref{eq:wSTg} that
\[w^{S,T}_G(\gm)\leq \frac{\prod_{u\in \partial V_\gamma}F_u}{\big(\sum_{i\in S} \lambda_i)^{|V_\gm^+\cap L|}\big(\sum_{j\in T} \lambda_j\big)^{|V_\gm^+\cap R|}},\]
Now for $u\in \partial V_\gm\cap L$, recall that $F_u=\sum_{i\in S}\lambda_i\prod_{v\in V_\gm\cap \partial u}B_{i,\sigma_\gm(v)}$. We have that there exist $i\in S$ and $v\in V_\gm\cap\partial u$ such that $B_{i,\sigma_\gm(v)}\leq \delta$; otherwise, since $\BB$ is a $\delta$-matrix, we would have that $B_{i,\sigma_\gm(v)}=1$ for all $i\in S$, and therefore $(S,T\cup \{\sigma_\gm(v)\})$ would be a biclique, contradicting the maximality of $(S,T)$ since $\sigma_\gm(v)\notin T$ (by the definition of $\mathcal{P}^{S,T}_G$). We therefore obtain that $F_u\leq \sum_{i\in S}\lambda_i -(1-\delta)\min(\lambdab)$. Similarly, for $u\in \partial V_\gm\cap R$, we have that $F_u\leq \sum_{j\in T}\lambda_j -(1-\delta)\min(\lambdab)$. Using the crude bounds $\min(\lambdab)\leq \sum_{i\in S}\lambda_i,\sum_{j\in T}\lambda_j\leq q$, we obtain that
\[w^{S,T}_G(\gm)\leq (\min(\lambdab))^{-|V_\gm|}\Big(1-\frac{(1-\delta)\min(\lambdab)}{q}\Big)^{-|\partial V_{\gm}|}.\]
By the $\rho$-extremality assumption (or, more precisely, by the definition of the set of polymers $\Pc^{S,T}_G$), we have that $|V_{\gm}|\leq 2q \rho n\leq \tfrac{n}{6\Delta}$, and therefore, by Lemma~\ref{lem:expansion}, $|\partial V_{\gm}|\geq \tfrac{\Delta}{7} V_{\gm}$. We therefore have that
\[w^{S,T}_G(\gm)\leq \emm^{-\big(\log(\min{\lambdab})+\tfrac{\Delta(1-\delta)\min(\lambdab)}{7q}\big)|V_{\gm}|}\]
Thus the polymer sampling condition is satisfied as long as $\tfrac{\Delta(1-\delta)\min(\lambdab)}{7q}+\log(\min(\lambdab))\geq 5+3\log((q-1)\Delta^3)$, which gives the desired conclusion.
\end{proof}

\section{Establishing phase maximality}\label{sec:maximal}

In this section we establish phase maximality for colorings and hard-core model. 
In particular, we prove Lemmas~\ref{lem:coloring} and \ref{lem:ind} from Section~\ref{subsec:our_approach}. Recall that the tree-recursion on the $\Delta$-regular tree for a general $q$-spin system with interaction matrix $\BB$ and activity vector $\lambdab$ is given by
\begin{equation}\tag{\ref{eq:tr}}
r_i \propto \lambda_i \left(\mbox{$\sum_{j\in [q]}$} B_{ij} c_j\right)^{\Delta-1} \mbox{ for $i \in [q]$};
\qquad
c_j \propto \lambda_j \left(\mbox{$\sum_{i\in [q]}$} B_{ij} r_i\right)^{\Delta-1},  \mbox{ for $j \in [q]$}.
\end{equation}
For the colorings and hard-core models, Lemma~\ref{lem:maxima} shows that the fixpoints of \eqref{eq:tr} include all maximizers of the function
\[
\Phi_{\BB,\lambdab,\Delta}(\rb,\cb)=  \frac{\rb^\TT \BB \cb}{\|\Lambdab^{-1}\rb\|_p \|\Lambdab^{-1}\cb\|_p}
\]
where $p=\frac{\Delta}{\Delta-1}$ and $\Lambdab$ is the diagonal matrix whose $i$-th diagonal entry is $\lambda_i^{1/\Delta}$. 
Finally, Corollary~\ref{lem:phases} implies that, to show $\tfrac{1}{12\Delta q}$-maximality for a set of maximal bicliques $\Bc_\Delta$, it is enough to show that
all maximizers $(\rb^*,\cb^*)$ of $\Phi_{\BB,\lambdab,\Delta}$ are Hessian dominant and satisfy that $\|(\alphab^*,\betab^*)-(\gb_S,\gb_T)\|_\infty\leq \tfrac{1}{15\Delta q}$ for some $(S,T)\in \Bc_\Delta$, where $(\alphab^*,\betab^*)=(\fb(\rb^*),\fb(\cb^*))$ is given by 
\[
\alpha^*_i= \frac{ ( \lambda_i^{-1/\Delta} r^*_i )^p }{\|\Lambdab^{-1}\rb^*\|_p^p} 
\quad\text{and}\quad
\beta^*_i= \frac{ ( \lambda_i^{-1/\Delta} c^*_i )^p }{\|\Lambdab^{-1}\cb^*\|_p^p} 
\quad\text{for~} i \in [q]. 
\]


\subsection{Phase maximality for colorings}
In this subsection we prove Lemma~\ref{lem:coloring} by showing phase maximality for colorings. 

Let $q,\Delta\geq 3$ be integers and let $d = \Delta - 1$. 
For $q$-colorings, using the correspondence in Example~\ref{example}, the tree-recursion can be written as:
\begin{equation}\label{eq:coloring-tr}
r_i = \frac{(1-c_i)^{\Delta-1}}{\sum_{j\in [q]}(1-c_j)^{\Delta-1}}, \quad c_i=\frac{(1-r_i)^{\Delta-1}}{\sum_{j\in [q]}(1-r_j)^{\Delta-1}} \quad\mbox{for~} i\in [q].
\end{equation}
Note that $(\gb_{[q]}, \gb_{[q]})$ is a trivial solution to \eqref{eq:coloring-tr}. 
The following lemma summarizes results from \cite[Section 7]{antiferro} and describes the nontrivial fixpoints of the tree recursion \eqref{eq:coloring-tr} when $q\ge4$ is an even integer in the non-uniqueness region $q<\Dt$. 

\begin{lemma}[\cite{antiferro}]
\label{lem:tree-dphase}
	Suppose that $q\ge 4$ is even and $\Dt > q$. 
	Then there is a one-to-one correspondence between all maximizers of $\Phi_{\BB,\onesb,\Delta}$ and all bicliques in $\Bc_\Delta=\{(S,[q]\backslash S) \, \big| \, |S|=\frac{q}{2}\}$:
	there exists $a=a(\Dt,q)$, $b=b(\Dt,q)$ satisfying $0< b < a < \frac{2}{q}$ and $a+b=\frac{2}{q}$, such that 
	every biclique $(S,[q] \setminus S)\in \Bc_\Delta$ corresponds to a maximizer $(\rb^*,\cb^*)$ of the form
	\[
		r^*_i = 
		\left\{
		\begin{aligned}
		a,\quad & i\in S;\\
		b,\quad & i\in [q] \setminus S,
		\end{aligned}
		\right.
		\quad\text{and}\quad
		c^*_i = 
		\left\{
		\begin{aligned}
		b,\quad & i\in S;\\
		a,\quad & i\in [q] \setminus S.
		\end{aligned}
		\right.
	\]
	Furthermore, all maximizers of $\Phi_{\BB,\onesb,\Delta}$ are Hessian dominant. 
\end{lemma}

We now prove Lemma~\ref{lem:coloring}, which we restate here for convenience.

\begin{lemmacoloring}
\statelemmacoloring
\end{lemmacoloring}



\begin{proof}
Let $k = \frac{q}{2}$ and $d = \Delta - 1$ for convenience. 
By Lemmas~\ref{lem:maxima} and \ref{lem:tree-dphase}, for a given biclique $(S, [q] \setminus S) \in \Bc_\Delta$, the corresponding maximizer $(\rb^*,\cb^*)$ of $\Phi_{\BB,\onesb,\Delta}$ satisfies the tree-recursion \eqref{eq:coloring-tr} as follows:
\[
	\left\{
	\begin{aligned}
	a &= \frac{\left( ka + (k-1)b \right)^d}{k \left[ \left( ka + (k-1)b \right)^d + \left( (k-1)a + kb \right)^d \right]};\\
	b &= \frac{\left( (k-1)a + kb \right)^d}{k \left[ \left( ka + (k-1)b \right)^d + \left( (k-1)a + kb \right)^d \right]},
	\end{aligned}
	\right.
\]
where $a,b$ are the constants given in Lemma~\ref{lem:tree-dphase}. 
We are going to show that, for sufficiently large $\Delta$, the constant $a$ is close to $\frac{2}{q}$ and the constant $b$ is close to $0$. 
Taking the ratio of $a$ and $b$, we get
\[
	\frac{a}{b} = \left( \frac{ka + (k-1)b}{(k-1)a + kb} \right)^d, \mbox{ and therefore } h = \left( \frac{h + t}{th + 1} \right)^d. 
\]
where $h = \frac{a}{b} > 1$ and $t = \frac{k-1}{k} = 1- \frac{2}{q}$. Consider the function $f(x) = \left( \frac{x+t}{tx+1} \right)^d$. 
Then $h$ is a fixpoint of $f$ (i.e., $f(h) = h$). 
In fact, the function $f$ has three fixpoints: $x = h>1$, $x = 1$, and $x = \frac{1}{h} < 1$. 
Let $h_0 = \emm^{\frac{d}{2q}} > 3$ when $d\ge 3q$.
We show next that $h > h_0$. 
By considering the monotone intervals of $f(x) - x$, it suffices to show that $f(h_0) > h_0$. 
We then compute that 
\begin{align*}
	\frac{1}{d} \log f(h_0) &= \log \left( \frac{h_0+t}{th_0+1} \right) 
	= \log \left( 1 + \frac{(1-t)(h_0-1)}{th_0+1} \right)\\
	&> \frac{(1-t)(h_0-1)}{2(th_0+1)} 
	> \frac{1-t}{4}
	= \frac{1}{2q},
\end{align*}
where the first inequality follows from $\log(1+\eps) > \frac{\eps}{2}$ for $\eps\in(0,1] $ and the second inequality is due to $h_0>3$ and $t<1$. 
Therefore, $f(h_0) > \emm^{\frac{d}{2q}} = h_0$ and thus $h > h_0$. 
	
Finally, notice that $(\alphab^*,\betab^*)=(\fb(\rb^*),\fb(\cb^*))$ is also in the form of
\[
		\alpha^*_i = 
		\left\{
		\begin{aligned}
		a',\quad & i\in S;\\
		b',\quad & i\in [q] \setminus S,
		\end{aligned}
		\right.
		\quad\text{and}\quad
		\beta^*_i = 
		\left\{
		\begin{aligned}
		b',\quad & i\in S;\\
		a',\quad & i\in [q] \setminus S,
		\end{aligned}
		\right.
\]
where $0< b'<a'<\frac{2}{q}$, $a'+b' = \frac{2}{q}$, and $\frac{a'}{b'} = \left( \frac{a}{b} \right)^p > \frac{a}{b} = h > h_0 = \emm^{\frac{d}{2q}}$.
It follows that
\[
	\norm{(\alphab^*,\betab^*)-(\gb_S,\gb_{[q] \setminus S})}_\infty = b' = \frac{2}{q (\frac{a'}{b'} + 1)} < \frac{2\emm^{-\frac{d}{2q}}}{q} \le \frac{1}{15 \Delta q}, 
\]
where the last inequality holds for $\Delta \ge 8 q \log \Delta$. 
The lemma then follows from Corollary~\ref{lem:phases}, using the fact from Lemma~\ref{lem:tree-dphase} that all maximizers are Hessian dominant. 
\end{proof}

Our proof of Lemma~\ref{lem:coloring} can also be modified to show that $O(\frac{1}{\Delta q})$-maximality fails when $q = \Omega(\frac{\Delta}{\log \Delta})$. 
This allows us to prove Lemma~\ref{lem:fail} from Section~\ref{subsec:our_approach}. 

\begin{proof}[Proof of Lemma~\ref{lem:fail}]
We use the same notation and approach as in the proof of Lemma~\ref{lem:coloring}. 
In particular, we show that $h < h_1$ for $h_1 = \emm^{\frac{4d}{q}}$, which can be deduced from $f(h_1) < h_1$. 
We have that
\[
\frac{1}{d} \log f(h_1) = \log \left( \frac{h_1+t}{th_1+1} \right) 
	= \log \left( 1 + \frac{(1-t)(h_1-1)}{th_1+1} \right)\\
	< \frac{(1-t)(h_1-1)}{th_1+1}
	\le \frac{1-t}{t}
	\le \frac{4}{q},
\]
where the last inequality is because $t = 1-\frac{2}{q} \ge \frac{1}{2}$. 
Therefore, we get $f(h_1) < \emm^{\frac{4d}{q}} = h_1$ and consequently $h < h_1$. 
It follows that $\frac{a'}{b'} = h^p < \emm^{\frac{4dp}{q}} = \emm^{\frac{4\Delta}{q}}$, 
and thus for $q \ge \frac{4\Delta}{\log \Delta}$ one has
\[
\norm{(\alphab^*,\betab^*)-(\gb_S,\gb_{[q] \setminus S})}_\infty = b' = \frac{2}{q (\frac{a'}{b'} + 1)} 
> \frac{\emm^{-\frac{4\Delta}{q}}}{q} \ge \frac{1}{\Delta q}. 
\]
Combining with Lemma~\ref{lem:hessiandominant}, this gives that $\frac{1}{2\Delta q}$-maximality fails when $q \ge \frac{4\Delta}{\log \Delta}$. 
\end{proof}

\subsection{Phase maximality for hard-core model}

In this subsection we consider the hard-core model and establishes phase maximality. The goal is to prove Lemma~\ref{lem:ind} from Section~\ref{subsec:our_approach}. 

Let $\Delta \ge 3$ be an integer and $\lambda > 0$ be a real. 
Recall from Example~\ref{example} that the interaction matrix $\BB=\{B_{ij}\}_{i,j\in\{0,1\}}$ for the hard-core model is given by $B_{00} = B_{01} = B_{10} = 1$ and $B_{11} = 0$, and the activity vector with fugacity $\lambdab=\{\lambda_i\}_{i\in\{0,1\}}$ is given by $\lambda_0 = 1$ and $\lambda_1 = \lambda$. 
Hence, the tree-recursion \eqref{eq:tr} becomes:
\[
r_1 = \frac{\lambda c_0^{\Delta-1}}{\lambda c_0^{\Delta-1} + 1}, \quad
r_0 = \frac{1}{\lambda c_0^{\Delta-1} + 1}, \quad
c_1 = \frac{\lambda r_0^{\Delta-1}}{\lambda r_0^{\Delta-1} + 1}, \quad
c_0 = \frac{1}{\lambda r_0^{\Delta-1} + 1}.
\]
As is standard, it would be easier to work with the ratios $x = \frac{r_1}{r_0}$ and $y = \frac{c_1}{c_0}$, so that the tree-recursion can be equivalently written as
\begin{equation}\label{eq:hc-tr}
x = \frac{\lambda}{(1+y)^{\Delta-1}}, \quad y = \frac{\lambda}{(1+x)^{\Delta-1}}.
\end{equation}
Note that the function $f(x) = \frac{\lambda}{(1+x)^{\Delta-1}}$ has a unique fixpoint $x_0$, and we are interested in the nontrivial solutions to \eqref{eq:hc-tr} (i.e., $(x,y) \neq (x_0,x_0)$). We restate Lemma~\ref{lem:ind} here for convenience. 
\begin{lemmaind}
\statelemmaind
\end{lemmaind}
\begin{proof}
Take an arbitrary maximizer $(\rb,\cb)$ of $\Phi_{\BB,\lambdab,\Delta}$ and let $x = r^*_1 / r^*_0$, $y = c^*_1 / c^*_0$. 
It is known that $x \neq x_0$, $y \neq x_0$, and $(\rb,\cb)$ is Hessian dominant when $\lambda > \lambda_c(\Delta)$ is in the non-uniqueness region; see, e.g., \cite{GSV-ising,antiferro}. Suppose that $x < y$ without loss of generality.

We first show that $x \le \frac{1}{30\lambda\Delta^2}$ when $\Delta \ge 50$ and $\lambda \ge \frac{50}{\Delta}$.
By \eqref{eq:hc-tr} we have 
\[
\lambda = x(1+y)^{\Delta-1} = y(1+x)^{\Delta-1}.
\]
Define $f(t) = \frac{(1+t)^{\Delta-1}}{t}$, and note that $f(x) = f(y)$. The function $f(t)$ is monotone decreasing when $t < \frac{1}{\Delta-2}$ and monotone increasing when $t > \frac{1}{\Delta-2}$. 
This implies $x < \frac{1}{\Delta-2} < y$. 
We then deduce from \eqref{eq:hc-tr} that for $\Delta \ge 50$ and $\lambda \ge \frac{50}{\Delta}$,
\[
y = \frac{\lambda}{(1+x)^{\Delta-1}} \ge \frac{\lambda}{( 1+\frac{1}{\Delta-2})^{\Delta-1}} 
\ge \frac{\lambda}{3} 
> \frac{1}{\Delta -2}. 
\]
Hence,
\[
f(x) = f(y) \ge f\left( \frac{\lambda}{3} \right) = \frac{3}{\lambda} \left( 1 + \frac{\lambda}{3} \right)^{\Delta-1}. 
\]
Meanwhile, we have 
\[
f\left( \frac{1}{30\lambda \Delta^2} \right) = 30\lambda \Delta^2 \left( 1 + \frac{1}{30\lambda \Delta^2} \right)^{\Delta-1}
\le 30\lambda \Delta^2 \emm^{\frac{1}{30\lambda \Delta}} \le 33 \lambda \Delta^2. 
\]
We claim that 
\begin{equation}\label{eq:toshow}
11 (\lambda \Delta)^2 \le \left( 1 + \frac{\lambda}{3} \right)^{\Delta-1}
\end{equation}
when $\Delta \ge 50$ and $\lambda \ge \frac{50}{\Delta}$. 
Given \eqref{eq:toshow}, we get
\[
f\left( \frac{1}{30\lambda \Delta^2} \right) \le 33 \lambda \Delta^2 \le \frac{3}{\lambda} \left( 1 + \frac{\lambda}{3} \right)^{\Delta-1} \le f(x)
\]
and thus $x \le \frac{1}{30\lambda \Delta^2}$ as wanted. 
It remains to prove \eqref{eq:toshow}. 
We consider two cases. 
If $\lambda \le 1$, then we have
\[
\frac{4}{3} \left( 1 + \frac{\lambda}{3} \right)^{\Delta-1}
\ge \left( 1 + \frac{\lambda}{3} \right)^\Delta
\ge \emm^{\frac{\lambda \Delta}{3+\lambda}} 
\ge \emm^{\frac{\lambda \Delta}{4}} 
\ge 15 (\lambda \Delta)^2,
\]
where the second inequality follows from $1+\eps \ge \exp(\frac{\eps}{1+\eps})$ for $\eps \in [0,1]$, and the last inequality holds when $\lambda \Delta \ge 50$. 
Meanwhile, if $\lambda > 1$ then we have
\[
\frac{9}{\lambda^2} \left( 1 + \frac{\lambda}{3} \right)^{\Delta-1}
\ge \left( 1 + \frac{\lambda}{3} \right)^{\Delta-3}
\ge \left( \frac{4}{3} \right)^{\Delta-3}
\ge 100 \Delta^2,
\]
where the last inequality holds when $\Delta \ge 50$. 
Therefore, \eqref{eq:toshow} holds when $\Delta \ge 50$ and $\lambda \ge \frac{50}{\Delta}$, and we conclude with $x \le \frac{1}{30\lambda\Delta^2}$ in this parameter regime.

Now, for a fixed $\Delta$, both the fixpoint $(\rb^*,\cb^*)$ of the tree recursion with $r^*_1 / r^*_0 < c^*_1 / c^*_0$ and the ground state $(\gb_0,\gb_{01})$ converge to $(1,0,0,1)$ as $\lambda$ tends to infinity.
Consequently, $(\alphab^*,\betab^*)$ converges to the same point as well. 
Hence, for $3\le \Delta < 50$, there exists a universal constant $C>0$ such that $\norm{(\alphab^*,\betab^*)-(\gb_0,\gb_{01})}_\infty \le \frac{1}{30 \Delta}$ whenever $\lambda \ge \frac{C}{\Delta}$. 
It remains to deal with the case that $\Delta \ge 50$ and $\lambda \ge \frac{50}{\Delta}$. 
Observe that
\[
\norm{(\alphab^*,\betab^*)-(\gb_0,\gb_{01})}_\infty = \max\left\{ \alpha^*_1, \frac{\lambda}{1+\lambda}-\beta^*_1 \right\}.
\]
We will upper bound the two terms using our bound on $x$. 
Recall that $p = \frac{\Delta}{\Delta-1}$.
First, we have
\[
\alpha^*_1 \le \frac{\alpha^*_1}{\alpha^*_0} = \left( \lambda^{-\frac{1}{\Delta}} \frac{r^*_1}{r^*_0} \right)^p = \lambda^{-\frac{1}{\Delta-1}} x^p \le 2 x \le \frac{1}{15 \lambda \Delta^2} \le \frac{1}{30\Delta}, 
\]
where $\lambda^{-\frac{1}{\Delta-1}} \le 2$ in our parameter regime. 
Next, notice that
\[
\frac{\lambda}{1+\lambda}-\beta^*_1 
= \frac{\lambda}{1+\lambda} - \frac{\beta^*_1 / \beta^*_0}{1 + \beta^*_1 / \beta^*_0}
= \frac{\lambda - \beta^*_1 / \beta^*_0}{(1+\lambda) ( 1 + \beta^*_1 / \beta^*_0 )}
\le \lambda - \frac{\beta^*_1}{\beta^*_0}. 
\]
Since we have
\[
\frac{1}{\lambda} \frac{\beta^*_1}{\beta^*_0} 
= \frac{1}{\lambda} \left( \lambda^{-\frac{1}{\Delta}} \frac{c^*_1}{c^*_0} \right)^p
= \left( \frac{y}{\lambda} \right)^p
= \frac{1}{(1+x)^\Delta},
\]
it follows that
\[
\frac{\lambda}{1+\lambda}-\beta^*_1  
\le \lambda \left( 1 - \frac{1}{(1+x)^\Delta} \right)
\le \lambda \left( 1 - \emm^{-\Delta x} \right)
\le \lambda\Delta x 
\le \frac{1}{30\Delta}. 
\]
This yields $\norm{(\alphab^*,\betab^*)-(\gb_0,\gb_{01})}_\infty \le \frac{1}{30\Delta}$. 
The lemma then follows from Corollary~\ref{lem:phases}. 
\end{proof}

\bibliographystyle{plain}
\bibliography{\jobname}

\end{document}